\documentclass[sn-mathphys,Numbered]{sn-jnl}

\usepackage{graphicx}%
\usepackage{multirow}%
\usepackage{amsmath,amssymb,amsfonts}%
\usepackage{amsthm}%
\usepackage{mathrsfs}%
\usepackage[title]{appendix}%
\usepackage{xcolor}%
\usepackage{textcomp}%
\usepackage{manyfoot}%
\usepackage{booktabs}%
\usepackage{algpseudocode}%
\usepackage{listings}%
\usepackage{amsmath}
\usepackage{amsthm}
\usepackage[linesnumbered,ruled,vlined]{algorithm2e}
\usepackage[colorinlistoftodos,prependcaption,textsize=tiny]{todonotes}
\usepackage{tikz}
\usetikzlibrary{shapes.geometric, arrows}
\tikzstyle{startstop} = [rectangle, rounded corners, minimum width=3cm, minimum height=0.6 cm,text centered, draw=black, fill=red!30]
\tikzstyle{io} = [trapezium, trapezium left angle=70, trapezium right angle=110, minimum width=3cm, minimum height=0.6 cm, text centered, draw=black, fill=blue!30]
\tikzstyle{io2} = [trapezium, trapezium left angle=70, trapezium right angle=110, minimum width=3cm, minimum height=0.6 cm, text centered, draw=black, fill=yellow!30]
\tikzstyle{io3} = [trapezium, trapezium left angle=70, trapezium right angle=110, minimum width=3cm, minimum height=0.6 cm, text centered, draw=black, fill=blue!30]
\tikzstyle{process} = [rectangle, minimum width=3.6cm, minimum height=0.6 cm, text centered, draw=black, fill=orange!30]
\tikzstyle{process2} = [rectangle, minimum width=3.6cm, minimum height=0.6 cm, text centered, draw=black, fill=red!30]
\tikzstyle{process3} = [rectangle, minimum width=3cm, minimum height=0.6 cm, text centered, draw=black, fill=green!30]
\tikzstyle{process4} = [rectangle, minimum width=3.6cm, minimum height=0.6 cm, text centered, draw=black, fill=blue!30]
\tikzstyle{process5} = [rectangle, minimum width=3.6cm, minimum height=0.6 cm, text centered, draw=black, fill=yellow!30]
\tikzstyle{process6} = [rectangle, minimum width=3.6cm, minimum height=0.6 cm, text centered, draw=black, fill=white!30]
\tikzstyle{decision} = [diamond, minimum width=2cm, minimum height=0.6 cm, text width=3cm, text centered, draw=black, fill=green!30]
\tikzstyle{block} = [rectangle, minimum width=0.6cm, minimum height=0.6cm, text centered, draw=black, fill=blue!30]
\tikzstyle{block2} = [rectangle, minimum width=0.6cm, minimum height=0.6cm, text centered, draw=black, fill=red!30]
\tikzstyle{block3} = [rectangle, minimum width=0.6cm, minimum height=0.6cm, text centered, draw=black, fill=green!30]
\tikzstyle{block4} = [rectangle, minimum width=0.6cm, minimum height=0.6 cm, text centered, draw=black, fill=yellow!30]
\tikzstyle{block5} = [rectangle, minimum width=0.6 cm, minimum height=0.6 cm, text centered, draw=black, fill=orange!30]
\tikzstyle{block6} = [rectangle, minimum width=0.6 cm, minimum height=0.6 cm, text centered, draw=black, fill=pink!30]
\tikzstyle{block7} = [rectangle, minimum width=0.6 cm, minimum height=0.6 cm, text centered, draw=black, fill=white!30]
\tikzstyle{process01} = [rectangle, minimum width=1.5cm, minimum height=0.6 cm, text centered, draw=black, fill=orange!30]
\tikzstyle{process21} = [rectangle, minimum width=1.5cm, minimum height=0.6 cm, text centered, draw=black, fill=red!30]
\tikzstyle{process31} = [rectangle, minimum width=1.5cm, minimum height=0.6 cm, text centered, draw=black, fill=green!30]
\tikzstyle{process41} = [rectangle, minimum width=1.5cm, minimum height=0.6 cm, text centered, draw=black, fill=blue!30]
\tikzstyle{process51} = [rectangle, minimum width=1.5cm, minimum height=0.6 cm, text centered, draw=black, fill=yellow!30]
\tikzstyle{arrow} = [thick,->,>=stealth]
\usepackage{adjustbox}
\usepackage{listings}
\usepackage{diagbox}
\usepackage{makecell}

\theoremstyle{thmstyleone}%
\newtheorem{theorem}{Theorem}
\theoremstyle{thmstyletwo}%
\newtheorem{example}{Example}%

\theoremstyle{thmstylethree}%
\newtheorem{definition}{Definition}%

\lstdefinestyle{mystyle}{
  basicstyle=\ttfamily\small,,
  breakatwhitespace=false,         
  breaklines=true,                 
  captionpos=b,                    
  keepspaces=true,                 
  numbers=left,                    
  numbersep=3pt,                  
  showspaces=false,                
  showstringspaces=false,
  showtabs=false,                  
  tabsize=2
}
\newcolumntype{?}{!{\vrule width 2pt}}
\usepackage[numbers]{natbib}
\usepackage{hyperref}
\usepackage{url}

\begin{document}

\title[Article Title]{Model Checking of \emph{vGOAL}}


\author{\fnm{Yi} \sur{Yang}}\email{yi.yang@kuleuven.be}

\author{\fnm{Tom} \sur{Holvoet}}\email{tom.holvoet@kuleuven.be}


\affil{\orgdiv{imec-DistriNet}, \orgname{KU Leuven}, \orgaddress{\street{Celestijnenlaan}, \city{Leuven}, \postcode{3001}, \country{Belgium}}}


\abstract{Developing autonomous decision-making requires safety assurance.
Agent programming languages like AgentSpeak and Gwendolen provide tools for programming autonomous decision-making. However, despite numerous efforts to apply model checking to these languages, challenges persist such as a faithful semantic mapping between agent programs and the generated models, efficient model generation, and efficient model checking.

As an extension of the agent programming language GOAL, \emph{vGOAL} has been proposed to formally specify autonomous decisions with an emphasis on safety. This paper tackles the mentioned challenges through two automated model-checking processes for \emph{vGOAL}:  one for Computation Tree Logic and another for Probabilistic Computation Tree Logic. Compared with the existing model-checking approaches of agent programming languages, it has three main advantages. First, it efficiently performs automated model-checking analysis for a given \emph{vGOAL} specification, including efficiently generating input models for NuSMV and Storm and leveraging these efficient model checkers. Second, the semantic equivalence is established for both nondeterministic models and probabilistic models of \emph{vGOAL}: from \emph{vGOAL} to transition systems or DTMCs. 
Third, an algorithm is proposed for efficiently detecting errors, which is particularly useful for \emph{vGOAL} specifications that describe complex scenarios. Validation and experiments in a real-world autonomous logistic system with three autonomous mobile robots illustrate both the efficiency and practical usability of the automated CTL and PCTL model-checking process for \emph{vGOAL}.}

\keywords{CTL Model Checking, PCTL Model Checking, \emph{vGOAL}, Autonomous Decision-Making, Agent Programming Languages}



\maketitle

\section{Introduction}
The advancement of technology is leading to an increased interest in systems capable of independently executing tasks, without the need for instructions or human intervention. At the core of these autonomous systems lies the autonomous decision-making.  The public rightfully expresses concerns about the safety of autonomous systems. Therefore, it poses a challenging yet indispensable endeavor to ensure the correctness of autonomous decision-making, particularly regarding safety.

Agent programming languages (APLs) have been extensively researched over many years as a means of programming autonomous agents within a multi-agent system \cite{shoham1993agent}. Throughout the past few decades, various APLs have emerged, including AgentSpeak \cite{bordini2005bdi}, Gwendolen \cite{dennis2008gwendolen}, and GOAL \cite{hindriks2009programming}. More recently, \emph{vGOAL} \cite{YY_EMAS23} has been proposed as a language, based on GOAL, for formally specifying autonomous decision-making. The interpreter of \emph{vGOAL} serves as an autonomous decision-maker that effectively generates safe decisions without requiring external formal verification for a sound \emph{vGOAL} specification \cite{YY_PAAMS23}. In this context, we consider a \emph{vGOAL} specification sound when it satisfies two properties: first, it can effectively achieve all goals; second, it always generates safe decisions under all realistic circumstances including encountering unexpected errors or expected errors.  

Formal verification is a convincing approach to ensure the correctness of agent programs. Model checking is the most successful and influential verification method in verifying APLs, including GOAL, Gwendolen, and AgentSpeak, owing to the automated verification process \cite{weiss_2013, bordini2003model, dennis2012model, jongmans2010model}. However, case studies for model checking of APLs remain relatively simple and conceptual due to three main challenges: a faithful semantic mapping between agent programs and generated models, efficient model generation, efficient model checkers \cite{bordini2003model, dennis2012model, jongmans2010model,dennis2018two}. Moreover,  autonomous systems often encompass multiple agents with numerous goals, the complexity can lead to extensive state spaces in their equivalent semantic models.

\begin{figure}[h]
   \centering
    \begin{tikzpicture}[node distance=2.2cm]
\node (1) [io2] {\emph{vGOAL} Specification};
\node (2) [process2, below of=1, yshift=+1 cm] {Model Generation};
\node (3) [io2, below of=2, xshift=-3.2cm,yshift=+1 cm] {Transition System};
\node (4) [io2, below of=2, xshift=+3.2cm,yshift=+1 cm] {DTMC};
\node (5) [process2, below of=3, yshift=+1 cm] {NuSMV Encoding};
\node (6) [process2, below of=4, yshift=+1 cm] {Storm Encoding};
\node (7) [io2, below of=5, yshift=+1 cm] {SMV File};
\node (8) [io2, below of=6, yshift=+1 cm] {PM File};
\node (9) [process2 , below of=7, yshift=+1 cm] {NuSMV};
\node (10) [process2 , below of=8, yshift=+1 cm] {Storm};
\node (11) [io2, below of=9, yshift=+1 cm] {CTL Model Checking Analysis};
\node (12) [io2, below of=10, yshift=+1 cm] {PCTL Model Checking Analysis};
\draw [arrow] (1) --(2);
\draw [arrow] (2) --(3);
\draw [arrow] (2) --(4);
\draw [arrow] (3) --(5);
\draw [arrow] (4) --(6);
\draw [arrow] (5) --(7);
\draw [arrow] (6) --(8);
\draw [arrow] (7) --(9);
\draw [arrow] (8) --(10);
\draw [arrow] (9) --(11);
\draw [arrow] (10) --(12);
\end{tikzpicture}
\caption{The Overview of the CTL and PCTL Model-Checking Process for \emph{vGOAL}.}
    \label{fig:AMC}
\end{figure}
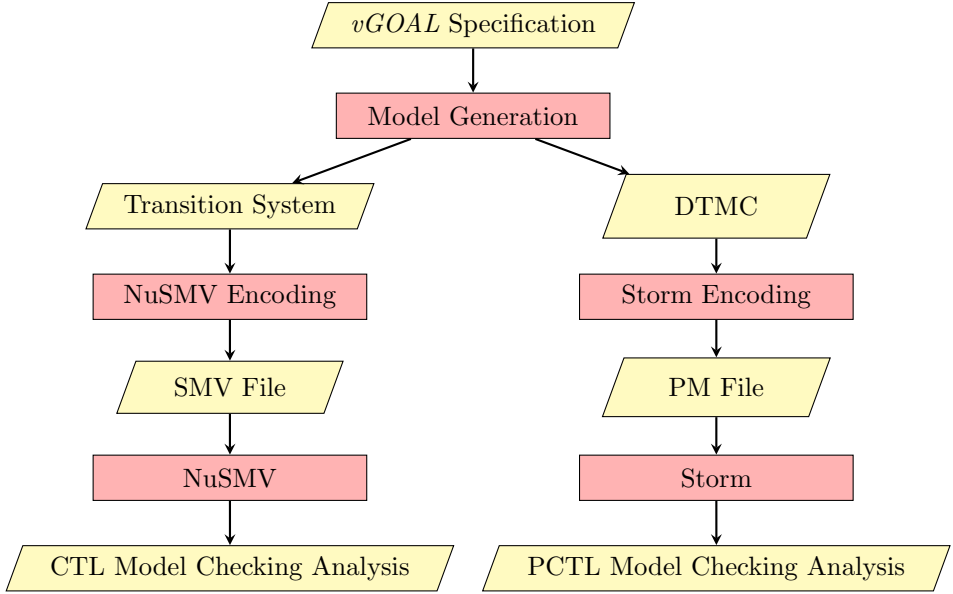

To address these challenges, we design and implement automated CTL and PCTL model-checking processes for \emph{vGOAL}, efficiently verifying autonomous decision-making. Figure \ref{fig:AMC} presents the overview of the automated CTL and PCTL model-checking processes for \emph{vGOAL}. The process takes a \emph{vGOAL} specification as input and produces both a CTL model-checking analysis and a PCTL model-checking analysis as outputs. Yellow parallelograms represent files during the model-checking processes. It can be a \emph{vGOAL} specification, a generated model: a transition system or a discrete-time Markov chain (DTMC), an input file for a model checker: an $smv$ file for NuSMV \cite{cimatti1999nusmv} or a $pm$ file for Storm \cite{dehnert2017storm}, or a CTL or PCTL mode-checking analysis. Red rectangles are steps of the model-checking process. The CTL model-checking process involves three steps: generating a transition system, encoding for NuSMV, and performing CTL model checking using NuSMV. The PCTL model-checking process involves three steps: generating a DTMC, encoding for Storm, and performing PCTL model checking using Storm.

The automated model-checking processes of \emph{vGOAL} addresses the aforementioned challenges well. First, the \emph{vGOAL} translator implements two equivalently semantic mappings from a \emph{vGOAL} specification to either a transition system or a DTMC. Second, the \emph{vGOAL} translator efficiently translates a \emph{vGOAL} specification to a transition system or DTMC. Third, the model-checking process for \emph{vGOAL} leverages the efficiency of NuSMV \cite{cimatti1999nusmv} and Storm \cite{dehnert2017storm}. Moreover, we propose an algorithm that can efficiently detect errors for a \emph{vGOAL} specification that violates safety or liveness properties if the autonomous system involves multiple agents with many goals. Using a real-world autonomous logistic system with three mobile robots, we have conducted experiments to illustrate the practical efficiency and the use of the automated model-checking processes for \emph{vGOAL}.

\section{Related Work}
The first significant work on verifying the belief–desire–intention (BDI) logic programming language through model checking of AgentSpeak(L) is detailed in \cite{bordini2003model}. To enable the model checking of AgentSpeak(L), two components are introduced: AgentSpeak(F)-a variant of AgentSpeak(L), and a translation from AgentSpeak(F) programs to Promela, the model specification language for the SPIN model checking \cite{holzmann1997model}. The model checking of AgentSpeak(F) has three main limitations. First, it verifies the properties of the abstract version of systems instead of the real implementation. Second, the translation process is difficult to understand \cite{dennis2012model}. Third, the model-checking process is not efficient. 

\cite{dennis2012model} describes the first comprehensive approach to the verification of programs developed using programming languages based on the BDI model of agency. The comprehensive approach is officially released as the MCAPL (Model-checking Agent Programming Languages) framework in \cite{dennis2018mcapl}. To verify various agent programming languages, the MCAPL consists of two layers, the Agent Infrastructure Layer (AIL) as a toolkit for interpreting BDI APLs, and Agent JavaPathFinder (AJPF) model checker as the automated verification tool. Interpreters have been implemented for Gwendolen, GOAL, SAAPL\cite{winikoff2007implementing}, and ORWELL \cite{dastani2009programming}. Compared with the model checking of  AgentSpeak(L), the MCAPL checks the executions of agent programs. However, the inefficiency of the AJPF model checker is particularly slow. It explores the whole state space for each property verification. For the same property verification, it took 65 seconds to verify in SPIN, while it took 9 hours in AJPF \cite{dennis2012model}. Moreover, the faithfulness of the interpretation of BDI language in AIL. To address the inefficiency issue, \cite{dennis2018two} presents a translation from AJPF to Promela, enabling the use of the SPIN model checker, and a translation from AJPF to PRISM, enabling the use of the PRISM model checker \cite{kwiatkowska2011prism}. The empirical comparison among AJPF, SPIN, and PRISM is conducted with rather simple case studies, the most complicated case study only has 408 states. The experiments show PRISM has significantly better efficiency and supports more expressive property specifications than AJPF and SPIN. However, the verification performance of Storm is still much better than the verification performance described in \cite{dennis2018two}. For a DTMC with 98 states, the verification time of PRISM took 1.3 seconds, while the verification time for Storm took only around 0.14 seconds for a DTMC with 138 states. 

\cite{jongmans2010model} presents an interpreter-based model checker (IMC) to model check agent programs. IMC includes a program interpreter to generate state space, and a model checker specifically built on this interpreter. This approach has only been validated using single-agent systems. Using a subset of GOAL as the target language, an empirical comparison among IMC, Maude model checker (MMC) \cite{eker2004maude}, and AJPF is conducted for a single deterministic agent system, specifically, blocks worlds. IMC has the best efficiency based on the conducted experiments. Compared with MMC and AJPF, IMC can verify single-agent systems with non-deterministic actions. Moreover, IMC uses no state-space reduction techniques, which makes it significantly more inefficient than advanced efficient model checkers. 

\cite{YY_EMAS22} describes an algorithm that transforms a single-agent and single-goal GOAL program to a semantically equivalent transition system. \cite{YY_23_Journal} extends the work described in \cite{YY_EMAS22}. It extends the transformation algorithm applicable for all single-agent systems, and it presents an automated model-checking process for single-agent GOAL programs using GOAL translator and Storm or PRISM.

The model-checking approaches described in \cite{bordini2003model,dennis2012model,dennis2018two,jongmans2010model,YY_23_Journal} have not specifically discussed how to handle many agent goals. 

\section{Preliminaries}
In this section, we introduce the basic concepts of \emph{vGOAL} specifications and explain the role of the goals of agents in \emph{vGOAL} semantics. Moreover, we briefly introduce the fundamentals of model checking. Since this introduction cannot be complete, we have to refer interested readers to \cite{YY_EMAS23} for more details of \emph{vGOAL}. For detailed insights into the foundational aspects of model checking, we recommend referring to \cite{baier2008principles}.

\begin{definition}(vGOAL Specifications) \cite{YY_EMAS23}\label{Def:vGOAL}\\
A vGOAL specification is defined as:
\begin{equation*}
\begin{aligned}
    vGOALSpec ::= & (MAS, K, C, A, S, P, E, D, Prob, Safety), \\
    MAS ::= & (id, B, goals, M_S, M_R)^*.
\end{aligned}
\end{equation*}
\end{definition}

A \emph{vGOAL} specification defines autonomous decision-making. The first main component is the agents' specifications, \( MAS \). Each agent's specification comprises a unique identifier (\( id \)), a belief base (\( B \)), a set of goals (\( goals \)), sent messages (\( M_S \)), and received messages (\( M_R \)).

The other specifications are system-wide:
$K$ : knowledge base, $C$: rules for enabled constraints generation,
$A$: rules for feasible action generation,
$S$: rules for sent message generation,
$P$: rules for event processing, including modifying agent goals and beliefs, and processing received messages,
$E$: action effects,
$D$: domain of all variables. 
Moreover, $Prob$ specifies the probability of possible outcomes of each action, i.e., it specifies the successful probability of each action. $Safety$ specifies all safety requirements the system has to satisfy.

\emph{vGOAL} syntax restricts it to first-order logic. All specifications, except \( E \) and \( D \), are interpreted either as a set of atoms or a first-order theory. The syntax guarantees each first-order theory has a minimal model. \( E \) is used to update the belief base through actions, and \( D \) is used to instantiate universally quantified variables.

\begin{definition} (vGOAL States)  \cite{YY_EMAS23} \label{Def:vGOALS}
A \emph{vGOAL} state is formalized as $state\text{::=} (substate)_{\times n}$, where $substate\text{::=}id\text{:}(I(B),I(goals))$.
\end{definition}

The \emph{vGOAL} $state$ of a system is formally defined as a composition of substates, $(id\text{:}(I(B),I(goals)))_{\times n}$. Each substate represents an agent with a unique identifier and the semantics of its beliefs and goals, denoted as $I(B)$ and $I(goals)$.

\begin{example} \label{ex:1}$MAS=[Agent_1,Agent_2]$\\
$Agent_1=(id_1,B_1,G_1,M_S,M_R)$, $Agent_2=(id_2,B_2,G_2,M_S,M_R)$\\
$id_1=A_1$, $B_1=[b_1,b_2,b_3]$, $G_1=[[g_1],[g_1],[g_2]]$\\
$id_2=A_2$, $B_2=[b_1,b_2,b_3]$, $G_2=[]$\\
$substate_1= A_1\text{:} (\{b_1,b_2,b_3\},\{g_1\})$, \\
$substate_2= A_2\text{:} (\{b_1,b_2,b_3\},\emptyset)$\\
$state=(substate_1,substate_2)$
\end{example}

Example \ref{ex:1} is a simple example of the \emph{vGOAL} state in a \emph{vGOAL} specification. A \emph{vGOAL} state only pertains to the agent's identifier and the interpretation of the agent's beliefs and goals. The interpretation of the agent's goals is the first goal base of the non-empty agent's goals or an empty set.
 
\begin{definition}(Operational Semantics of vGOAL) \cite{YY2024} \label{Def:OS}\\
The operational semantics of \emph{vGOAL} are defined by:
\begin{equation*}
(substate)_{\times n} \xrightarrow{Action} (substate')_{\times n},
\end{equation*}
where Action is defined in two cases:
\begin{itemize}
\item Non-deterministic: \(Action::= (id : (M_R, Act))_{\times n} \),
    \item Probabilistic: \(Action::= (id : (M_R, Act, \mathbb{P}(M_R, Act)))_{\times n} \).
\end{itemize}
\end{definition}

A \emph{vGOAL} state will be updated if any substate is updated. A substate will only be changed by either the received messages of the agent, an agent action, or both of them. Moreover, the selection of actions can be modeled as nondeterministic or probabilistic. Hence, $Action$ can be formalized in two ways. In the non-deterministic case, actions are non-deterministically chosen from a set of all possible actions. In the probabilistic case, actions are selected according to the probability distribution $P(M_R, Act)$.

\begin{definition}(Transition System) \cite{baier2008principles}\label{Def:TS}
A transition system TS is a tuple $(S, Act, \rightarrow, S_0, F, AP, L)$ where $S$ is a set of states, Act is a set of actions, $\rightarrow \subseteq S \times Act \times S$ is a transition relation, $S_0\subseteq S$ is a set of initial states, $F \subseteq S$ is a set of final states, $AP$ is a set of atomic propositions and $L: S \rightarrow 2^{AP}$ is a labeling function.
\end{definition}
\begin{definition}
(Path) \cite{baier2008principles} \label{Def: Path} A finite path fragment $\pi$ of $TS$ is a finite state sequence $s_0s_1...s_n$ such that $\forall 0\le i\le n. s_i\in Post(s_{i-1})$
, where $n\ge 0$. A maximal path fragment is a finite path fragment that ends in a terminal state. A path fragment is called initial if it starts in an initial state. A path of transition system $TS$ is an initial, maximal path fragment.
\end{definition}

\begin{definition} (Bisimulation Equivalence) \cite{baier2008principles}\label{Def:Bieq} $TS_i=(S_i, Act_i, \rightarrow_{i},S_{i0}, AP_i, L_i)$, $i=1,2$, be transition systems over $AP$. A bisimulation for $(TS_1,TS_2)$ is a binary relation $R\subseteq S_1\times S_2$, such that 
(a) $\forall s_1\in S_{10} (\exists s_2\in S_{20}. (s_1,s_2)\in R)$, and $\forall s_2\in S_{20} (\exists s_1\in S_{10}. (s_2,s_1)\in R)$, (b) for all $(s_1,s_2)\in R$, it holds (b.1) $L_1(s_1)=L_2(s_2)$; (b.2) if $s_1'\in Post(s_1)$, then there exist $s_2'\in Post(s_2)$ with $(s_1',s_2')\in R$; (b.3) if $s_2'\in Post(s_2)$, then there exist $s_1'\in Post(s_1)$ with $(s_2',s_1')\in R$.
\end{definition}

\begin{definition} (Satisfaction Relation for CTL) \cite{baier2008principles} \label{Def:CTL} Let $a\in AP$ be an atomic proposition, $TS=(S, Act, \rightarrow, S_0, F, AP, L)$ be a transition system, state $s\in S$, $\phi$ be a CTL path formula, such that (a) $s\models a$ iff $a\in L(s)$; (b) $s\models \exists \phi$ iff $\pi \models \phi$ for some $\pi \in Paths(s)$; (c) $s\models \forall \phi$ iff $\pi \models \phi$ for all $\pi \in Paths(s)$.  
\end{definition}

\begin{definition}(Discrete-Time Markov Chain) \cite{baier2008principles}\label{Def:DTMC}
A Discrete-Time Markov Chain (DTMC) is a tuple $D=(S,P, \iota_{init}, AP, L)$ where
S is a countable, nonempty set of states, $P: S \times S \rightarrow [0,1]$ is the transition probability function such that for all states s: $\sum_{s'\in S}P(s,s')=1$, $\iota_{init}:S\rightarrow [0,1]$ is the initial distribution, such that 
$\sum_{s\in S}\iota_{init}(s)=1$, and $AP$ is the set of atomic propositions and $L:S\rightarrow 2^{AP}$ is a labeling function. 
\end{definition}

\begin{definition} (Satisfaction Relation for PCTL) \cite{baier2008principles}\label{Def:PCTL} Let $a\in AP$ be an atomic proposition, $D=(S,P, \iota_{init}, AP, L)$ be a Markov chain, state $s\in S$, and $\phi$ be a PCTL path formula, such that (a) $s\models a$ iff $a\in L(s)$; (b) $s\models P_J(\phi)$ iff $Pr(s\models \phi)\in J$, where $J\in [0,1]$ and $Pr(s\models \phi)=Pr_s\{\pi \in Paths(s)|\pi \models \phi\}$.
\end{definition}

\begin{definition} (Probabilistic Bisimulation) \cite{baier2008principles}\label{Def:probBieq} Let $D=(S,P, \iota_{init}, AP, L)$ be a Markov chain. A probabilistic bisimulation on $D$ is an equivalence relation $R$ on $S$ such that for all states $(s_1, s_2)\in R$: (a) $L(s_1)=L(s_2)$, (b) $P(s_1, T)=P(s_2,T)$ for each equivalence class $T\in s/R$. States $s_1$ and $s_2$ are bisimulation-equivalent, denoted $s_1\sim_{D} s_2$, if there is a bisimulation $R$ on $D$ such taht $(s_1,s_2)\in R$.
\end{definition}

\section{Model Generation}
This section presents how to generate two faithful semantic mappings from a \emph{vGOAL} specification to its equivalent models. Following Definition \ref{Def:OS}, for a given \emph{vGOAL} specification,  the non-deterministic model is a transition system and the probabilistic model is a DTMC. 
First, we introduce two algorithms for the model generation of a \emph{vGOAL} specification. The first algorithm describes how to generate a semantically equivalent transition system from a given \emph{vGOAL} specification. The second algorithm describes how to generate a semantically equivalent DTMC from a given \emph{vGOAL} specification. Second, we provide theoretical proof of the semantic equivalence between a given \emph{vGOAL} specification and its models.  Additionally, the two algorithms are implemented in the \emph{vGOAL} translator, and its source code is available at \cite{vGOALsource}.

\subsection{Algorithms}
\subsubsection{Transition System Generation}
Algorithm \ref{Alg:1} describes the transformation from a \emph{vGOAL} specification for an autonomous system with $n$ agents to its semantically equivalent transition system. In this context, both a \emph{vGOAL} specification and a transition system adhere to the formalizations outlined in Definition \ref{Def:vGOAL} and Definition \ref{Def:TS}, respectively.Based on Definition \ref{Def:OS}, the semantics inherent in \emph{vGOAL} is naturally suitable to construct a transition system. The \emph{vGOAL} interpreter implements the semantics of \emph{vGOAL}. Leveraging the existing \emph{vGOAL} semantics implementation, we construct the transition system corresponding to a given \emph{vGOAL} specification.

\begin{algorithm}
  \caption{ Transform a \emph{vGOAL} specification to a transition system}
  \label{Alg:1}
\LinesNumbered
\SetKwInput{KwInput}{Input}                
\SetKwInput{KwOutput}{Output}              
\DontPrintSemicolon
  \KwInput{$vGOALSpec::= (MAS,K, C, A,S, P,E, D, Prob, Safety)$, $MAS::=(id, B, goals, M_S, M_R)^n$}
   \KwOutput{$(S,Act,T, S_0, F, AP, L)$}
  $S_0\leftarrow \bigcup_{id=1}^{n} \{id:(I(B),I(goals))\}$,
  $S\leftarrow \{S_0\}$,$F\leftarrow \emptyset$\; $D_{T}\leftarrow \emptyset$, 
 $S_{cur}\leftarrow \{S_0\}$, $Act\leftarrow \emptyset$,$T\leftarrow \emptyset$,$AP\leftarrow \emptyset$, $L\leftarrow \emptyset$\;
  \While{$S_{cur}\ne \emptyset$}{
   $S_{next}\leftarrow \emptyset$\;
   \For{each $\bigcup_{id=1}^{n} \{id: (I(B),I(goals))\}\in S_{next}$}{
     \If{$\forall id. I(goals)=\emptyset$}{
     $F\leftarrow F\cup \{s\}$
     }
     }
     $S_{cur}\leftarrow S_{cur} \setminus F$\;
   \For{each $s \in S_{cur}$}{
   \If{$s\notin D_T.keys()$ }{
  $L\leftarrow L\cup \{s:\emptyset\}$\;
   \For{each $Ag: (id,B,goals,M_S,M_R)\in MAS$}{
   $subS'_{id}\leftarrow \emptyset$, $subA_{id}\leftarrow \emptyset$, $subT_{id}\leftarrow \emptyset$\;
   $L[s]\leftarrow L[s] \cup \{id: MinModel(B,KB)\}$\;
   $AP\leftarrow AP \cup MinModel(B,KB)$\;
   $subs\leftarrow id:(B,I(goals))$\;
   $(subS_{id}',subA_{id}) \leftarrow  DM(Ag,K,C,A,S,P,D)$\;
   $subS'\leftarrow Update(subS',subA,E)$\;
   \For{each $(suba,subs')\in (subA_{id},subS'_{id})$ }{$subT\leftarrow subT\cup \{(subs,suba, subs')\}$
   }
   }
   $Ag\leftarrow Communication(Ag)$\;
   $S_n\leftarrow \Pi_{id=1}^n subS'_{id}$, $S\leftarrow S\cup S_n$ \;$Act_n\leftarrow \Pi_{id=1}^n subA'_{id}$, $Act\leftarrow Act\cup Act_n$ 
   \;$T_n\leftarrow \Pi_{id=1}^n subT_{id}$, $T\leftarrow T\cup T_n$ \;  
     $D_T\leftarrow D_T\cup \{s:S_n\}$,$S_{next}\leftarrow S_{next}\cup S_n$ 
   }
   \Else{
   $states\leftarrow D_T[s]$, $D_T[s]\leftarrow \emptyset$\;
   \For{each $state \in states$}{
   $s'\leftarrow \emptyset$\;
\For{$subs::=id:(I(B),I(goals))\in s'$}{
   \If{$I(goals)$ is a original goal}{
   $s'\leftarrow s'\cup \{id:(I(B),I(goals[1:]))\}$
   }
   \Else{
   $s'\leftarrow s' \cup \{subs\}$
   }
}
$S_{next}\leftarrow S_{next} \cup \{s'\}$\; $D_T[s]\leftarrow D_T[s]\cup \{s'\}$\;
\If{$s'\notin S$}{
$S\leftarrow S \cup \{s'\}$, $L\leftarrow L\cup \{s':L[s]\}$\;
$T\leftarrow T \cup \{(s, Act(s,state),s' )\}$\;
 
 $\forall id. M_R\leftarrow \emptyset$ 
}

   }
   }
    
   }
   $S_{cur}\leftarrow S_{next}$
   }
   
\KwRet $(S, Act, T, S_0, F, AP, L)$
\end{algorithm}
The initial state ($S_0$) of the transition system is determined by the semantics of the initial beliefs and goals of each agent within the autonomous systems. The transition system begins with $S$ initialized to this initial state.
Final states ($F$) are defined as states in which all agents within the autonomous system have no goals.
$D_T$ functions as a dictionary tasked with recording the current state and all its possible next states in the last reasoning cycle.
$S_{cur}$ denotes all possible states during a cycle when generating the transition system, initialized with the initial state.
$Act$ denotes all actions and communication.
$T$ is the set of all transitions.
$AP$ denotes all atomic propositions generated as state properties.
$L$ is a dictionary that serves to record state properties. $F$, $D_T$, $Act$, $T$, $AP$, and $L$ are initialized with an empty set.

Lines 3-41 describe an iterative generation process of a transition system. The generation process will terminate if all possible executions starting from the initial state are traversed. $S_{next}$ records all possible next states of all states in $S_{cur}$.Lines 5–7 describe the generation of final states. A state attains finality when each substate is devoid of goals, which corresponds to no goals for each agent.  The transition system generation process will explore all possible next states for each non-terminal state for all current states. 

The next states of each state can be calculated by 
the \emph{vGOAL} decision-making process, implemented in the \emph{vGOAL} interpreter. However, Algorithm \ref{Alg:1} only uses the \emph{vGOAL} decision-making process to calculate next states when the state has never been processed by the \emph{vGOAL} decision-making process. If the \emph{vGOAL} decision-making process has processed the state, Algorithm \ref{Alg:1} will generate the next states based on the generated next states. Lines 10–25 describe the process of generating the next state using the \emph{vGOAL} decision-making process. Specifically, Lines 11–21 detail the implementation of \emph{vGOAL} semantics, the identical implementation of the \emph{vGOAL} interpreter. According to Definition \ref{Def:OS}, agents execute actions in a modular manner, consequently updating the corresponding substate and subsequently modifying the overall state. $subS'_{id}$ records all potential next substates of $Agent_{id}$ and starts as an empty set. Simultaneously, $subA_{id}$ records all possible actions of $Agent_{id}$, and $subT_{id}$ records all possible transitions of $Agent_{id}$. Lines 26-40 outline the generation of the next states based on the previously generated next states. The \emph{vGOAL} decision-making process generates decisions based on current beliefs and the current focused goal of each agent. Consequently, the generated decisions should remain consistent with those produced previously. Deviations in the next states are only possible when transitioning from one focused goal to the next.

\subsubsection{DTMC Generation}
\begin{algorithm}
  \caption{Transform a \emph{vGOAL} specification to a DTMC in Storm}
  \label{Alg:2}
\LinesNumbered
\SetKwInput{KwInput}{Input}                
\SetKwInput{KwOutput}{Output}              
\DontPrintSemicolon
\KwInput{$vGOALSpec::= (MAS,K, C, A,S, P,E, D, Prob, Safety)$, \\$(S_1,Act,T, S_0, F, AP_{1}, L)$}
  \KwOutput{$(S_2, P, \iota_{init},AP_{2})$}
  $S_2\leftarrow S_1$, $S_0'\leftarrow S_0$, $AP_{2}\leftarrow AP_1$,  $L_2\leftarrow L_1$, $P\leftarrow \{\}$\;
  \For{each $(s_1 \times Acts \rightarrow s_2)  \in T$}{
   $p\leftarrow 1.0$\;
   \For{each $act \in Acts$}{
   \If{$act \in Prob.keys()$}{
   $p\leftarrow p* Prob[act]$
   }
   \If{$s_1 \in P.keys()$}{
   $P[s_1].append(p,s_2)$
   }
   \Else{
   $P.update(\{s_1: (p,s_2)\})$
   }
   $P\leftarrow P\cup \{s_1, p,s_2\}$
   }

  }
  $sum\leftarrow 0$\;
   \For{each $s \in P.keys()$}{
   \For{each $(p,s') \in P[s]]$}{
   $sum\leftarrow sum+p$
   }
   \If{$sum\ne 1$}{
   \For{each $(p,s') \in P[s]]$}{
   $p\leftarrow p/sum$
   }
   }
   }
   $\iota_{init}\leftarrow P[S_0']$\;
\KwRet $(S_2,P, \iota_{init},AP_2,L_2)
 $
\end{algorithm}
Algorithm \ref{Alg:2} describes how to transform a \emph{vGOAL} specification to a DTMC in Storm. Following Definition \ref{Def:OS}, the main difference between the non-deterministic model and the probabilistic model is the model of uncertainty. Therefore, we directly use the generated output of Algorithm \ref{Alg:1} to avoid unnecessary computation.

The inputs of Algorithm \ref{Alg:2} consist of the \emph{vGOAL} specification and the generated transition system of Algorithm \ref{Alg:1}. The output of Algorithm \ref{Alg:2} is the equivalent probabilistic model of the given \emph{vGOAL} specification. Following Definition \ref{Def:TS} and Definition \ref{Def:DTMC}, a transition system is a tuple $(S_1, Act, \rightarrow, S_0,F,AP_1,L_1)$, and a DTMC is a tuple $(S_2,P, \iota_{init}, AP_2, L_2)$. We can easily construct bijections from $S_1$ to $S_2$, from $S_0$ to $S_0'$, from $AP_1$ to $AP_2$, and from $L_1$ to $L_2$. 

The most important part of Algorithm \ref{Alg:2} is the probabilistic model of uncertainty, which involves $P$ and $\iota_{init}$. $P$ is initialized with an empty set. Lines 2-11 describe how to compute the probability of each transition using the probabilistic specification of each action ($Prob$). Notably, if an action has not been specified in $Prob$, the probability is calculated as 1. Lines 12-18 describe the adjustment of the probability of each transition, ensuring the sum of the probabilistic distribution of a state is 1. $\iota_{init}$ is obtained by computing the initial probabilistic distribution using $P$ and $S_0'$.

\subsection{Theoretical Proofs}
\begin{theorem}\label{thm:Correctness}
Algorithm \ref{Alg:1} generates an equivalent transition system for a given \emph{vGOAL} specification of $n$-agent autonomous system in terms of operational semantics.
\end{theorem} 
\begin{proof} Following Definition \ref{Def:vGOAL}, the given \emph{vGOAL} specification is defined as $vGOALSpec::=(MAS,K, C, A,S, P,E, D, Prob, Safety)$ and $MAS::=(id, B, goals, M_S, M_R)^n$, we denote the semantics of $vGOALSpec$ as $TS_1$; we denote the transition system generated by Algorithm 1 as $TS_2$. Following Definition \ref{Def:TS}, $TS_1$ and $TS_2$ are in the form $TS_i::=(S_i, Act_i, \rightarrow_{i}, S_{i0}, F_i, AP_i, L_i), i\in \{1,2\}$. We show $TS_1$ is bisimulation-equivalent to $TS_2$, denoted $TS_1 \sim TS_2$, and $\sim$ is an equivalence relation, denoted $R_{id}$,therefore, $TS_1 \equiv TS_2$. 

Following Definition \ref{Def:Bieq}, we prove the condition $a$ of the bisimulation equivalence. Following Definition \ref{Def:vGOALS}, the initial state of the given \emph{vGOAL} specification is the semantics of each agent's beliefs and goals, denoted as $S_{10}=\bigcup_{id=1}^n\{id: (I(B), I(goals))\}$. According to Line 1 in Algorithm \ref{Alg:1}, $S_{20}=\bigcup_{id=1}^n\{id: (I(B), I(goals))\}$. As $S_{10}=S_{2,0}$, $(S_{10},S_{20})\in R_{id}$ holds for the initial state.

Second, we show the bisimulation with equivalence relation on paths. $Paths(s_i)(i\in\{1,2\})$ follows Definition \ref{Def: Path}. For each finite path $\pi_1=s_{10}s_{11}s_{12}...\in Paths(s_1)$, there exists a path $\pi_2=s_{20}s_{21}s_{22}...\in Paths(s_2)$ of the same length such that $(s_{1k},s_{2k})\in R_{id}$ for all $k$. Following Definition \ref{Def:vGOALS}, $s_{1k}::=\bigcup_{id=1}^n\{id: (I(B_k),$ $I(goals_k))\}$. Following Definition \ref{Def:Bieq}, we prove the condition $b.2$ by induction on k.  For each case, we distinguish between $s_i$ being a terminal state or not. \\
\textbf{Base case ($k=0$):} $s_1=s_{10}$, $s_2=s_{20}$,
$(s_1,s_2)\in R_{id}$ \\
\textbf{Case 1:} In case $s_1$ is a terminal state, that means all agents of the autonomous system defined by $vGOALSpec$ have no goals to achieve. In this case, \emph{vGOAL} interpreter will not generate any decisions. Algorithm \ref{Alg:1} will also not generate any decisions for $s_2$, because the state will be considered as a final state of $TS_2$. \\
\textbf{Case 2:} In case $s_1$ is not a terminal state, 
$s_{1}$ is updated by the \emph{vGOAL} semantics, $s_{10}\rightarrow s_{11}$. As $s_2=s_{20}$ is an initial state, it has goals to achieve. Algorithm 1 will also generate the next state of $s_{2}$ using the \emph{vGOAL} semantics. Therefore, $s_{20}\rightarrow s_{21}$, $(s_{11}, s_{21})\in R_{id}$, and $s_{11}=s_{21}$.  \\
\textbf{Induction step:} Assume $k\ge 0$, and the path fragment $s_{20}s_{21}...s_{2k}$ is already constructed with $\forall 0\le i\le k. (s_{1i},s_{2i})\in R_{id}$.\\
\textbf{Case 1:} In case $s_{1k}$ is a terminal state, we have $\forall id. I(goals_k)=\emptyset$, $s_{2k}$ will also be a final state according to Lines 6-7 in Algorithm \ref{Alg:1}.\\
\textbf{Case 2:} In case $s_{1k}$ is not a terminal state, we have $\exists id. I(goals_k)\ne \emptyset$.   $\pi_1$ is extended by $s_{1k}s_{1(k+1)}$ based on the \emph{vGOAL} semantics. As $(s_{1k},s_{2k})\in R_{id}$, $s_{2k}$ is a non-final state of $TS_2$. Algorithm \ref{Alg:1} generates the next state of $s_{2k}$ based on two cases.\\
\textbf{Case 2a:} If $(s_{2k},s_{2(k'+1)})\notin D_T$,
which means $s_{2k}$ has never been processed by the implementation of the \emph{vGOAL} semantics during the transition system generation, it will be processed by the implementation of the \emph{vGOAL} semantics, described by lines 10-25. $\pi_2$ is extended by $s_{2k}s_{2(k+1)}$ based on \emph{vGOAL} semantics. Therefore, $(s_{1(k+1)}, s_{2(k+1)}\in R_{id}$. \\
\textbf{Case 2b:} If $(s_{2k},s_{2(k'+1)})\in D_T$,
$s_{2k}$ has been processed by the \emph{vGOAL} semantics during the transition system generation. Algorithm \ref{Alg:1} will generate the next state based on the previous generation results, described by lines 26-40. The previous generation certainly produced a path fragment with $s_{2k'}s_{2(k'+1)}$, where $s_{2k}=s_{2k'}$. The \emph{vGOAL} semantics is updated only based on the state for one \emph{vGOAL} specification. The decision changes the beliefs, which possibly achieves the current goal, subsequently shifting the focused goal to the next goal. $s_{2(k+1)}$ and $s_{2(k'+1)}$ will be the same except the current focused goal of any agent is shifted to the next.
Therefore, it is necessary to modify all agent goals to get the right next state of $s_{2k}$. We denote a substate of $s_{1k}$ for agent $i$ as $sub_{1(k)i}$, where $1\le i\le n$. The sub path fragment is denoted as $sub_{1(k)i}sub_{1(k+1)i}$, we prove $(s_{1(k+1)},s_{2(k+1)})\in R_{id}$ by proving $\forall 1\le i\le n. (sub_{1(k+1)i},sub_{2(k+1)i})\in R_{id}$.\\
\textbf{Case 2b.1:} If $I(goals_{1(k+1)i}))\notin I(goals_{10i})$, the generated decision does not result in the shift of the current focus goal. Based on Lines 33-34, $sub_{2(k+1)i}=sub_{2(k'+1)i}$, the substate update only depends on the \emph{vGOAL} semantics. $(sub_{1(k+1)i},sub_{2(k+1)i})\in R_{id}$ holds for this case.\\
\textbf{Case 2b.2:} If $I(goals_{1(k+1)i}))\in I(goals_{10i})$, the generated decision results in the shift of the current focused goal. Based on Lines 31-32, $sub_{2(k+1)i}=id: (I(B_{2(k+1)}), I(goals_{2(k'+1)}[1:]))$. This modification calibration with the \emph{vGOAL} semantics update for the two states, $s_{2k}$ and $s_{2k'}$ where $I(B_{2k})=I(B_{2k'})$, $I(goals_{2k})=I(goals_{2k'})$, and $goals_{2k}\ne goals_{2k'}$.

By symmetry, for each path $\pi_2\in Paths(s_2)$, there exists a path $\pi_1\in Paths(s_1)$ of the same length which is statewise related to $\pi_2$. Therefore, following Definition \ref{Def:Bieq}, the condition $b.3$ of the bisimulation equivalence is proven. Moreover, the state properties of $TS_1$ and $TS_2$ are generated by \emph{vGOAL} semantics updates, and the condition $b.1$ of the bisimulation equivalence is proven.
\end{proof}

\begin{theorem}\label{thm:Correctness2}
Algorithm \ref{Alg:2} generates an equivalent DTMC for a given \emph{vGOAL} specification of $n$-agent autonomous system in terms of operational semantics.
\end{theorem} 
\begin{proof}
Following Definition \ref{Def:vGOAL}, the given \emph{vGOAL} specification is defined as $vGOALSpec::=(MAS,K, C, A,S, P,E, D, Prob, Safety)$ and $MAS::=(id, B, goals, M_S, M_R)^n$, we denote the probabilistic semantics of $vGOALSpec$ as $D_1$; we denote the DTMC generated by Algorithm 1 as $D_2$. Following Definition \ref{Def:DTMC}, $D_1$ and $D_2$ are in the form $D_i::=(S_i,P, \iota_{init}, AP_i, L_i), i\in \{1,2\}$. We show $D_1$ is probabilistic bisimulation-equivalent to $D_2$, denoted $D_1 \sim_{D} D_2$, and $\sim_{D}$ is an equivalence relation for Markov chains, denoted $R_{id}$,therefore, $D_1 \equiv D_2$. \\
Algorithm \ref{Alg:2} computes the DTMC based on the generated transition system of Algorithm \ref{Alg:1}. Therefore, $D_i$ and $TS_i$ ($i=\{1, 2\}$) are equivalent except the model of uncertainty. According to Theorem \ref{thm:Correctness}, $TS_1\equiv TS_2$, $D_1$ and $D_2$ are equivalent on  $(S_i,AP_i, L_i), i\in \{1,2\}$. Hence, to show $D_1\equiv D_2$, we only need to show the probabilistic bisimulation between $D_1$ and $D_2$. \\
Following Definition \ref{Def:probBieq}, we need to prove condition $a$ and condition $b$ to show the equivalence. Each state of the \emph{vGOAL} specification has a bijective state of $D_2$. Consequently, condition $a$ is proven. Each states in $D_1$ have a bijective state in $D_2$. Therefore, each probabilistic transition in $D_1$ has a bijective probabilistic transition in $D_2$. The form of each probabilistic transition in $D_1$ is $(s_1 , p, s_2)$, which has a corresponding transition in $TS_1$ in the form $(s_1 , Acts, s_2)$. $p$ is determined by the $Acts$ and $Prob$. For each transition in $TS_1$, there is a corresponding equivalent transition in $TS_2$. Each probabilistic transition in $D_2$ is determined by its corresponding transition in $TS_2$ and the probability specifications of each action.  Hence the condition $b$ is proven.
\end{proof}

\section{Quick Error Detection}
This section presents two algorithms that are used to quickly detect errors for a \emph{vGOAL} specification, which involves many agents and each agent has multiple goals. The hypothesis of the quick detection algorithms is a finite set of different goals for each agent. 
\subsection{Algorithms}
\begin{algorithm}
  \caption{Reduce a \emph{vGOAL} specification for a one-agent system}
  \label{Alg:3}
\LinesNumbered
\SetKwInput{KwInput}{Input}                
\SetKwInput{KwOutput}{Output}              
\DontPrintSemicolon
\KwInput{$vGOALSpec_1::= (MAS_1,K, C, A,S, P,E, D)$, $MAS::=(id, B, goals_1, M_S, M_R)$, $goals_1::= goal_1 goal_2...goal_{n-1}goal_n$, $Goals::=\{G_1,...,G_m\}$ $\forall 1\le j\le n. goal_j\in Goals$, $1\le n$, $1\le m$}
  \KwOutput{$vGOALSpec_2::= (MAS_2,K, C, A,S, P,E, D)$, $MAS_2::=(id, B, goals_2, M_S, M_R)$}
$vGOALSpec_2\leftarrow vGOALSpec_1$,
$goals_2\leftarrow goal_1$\;
$goals_1\leftarrow goal_2...goal_{n-1}$,
$D_{G}\leftarrow \emptyset$\;
\For{each $goal_i\in goals_1$}{
\If{$Goals\ne D_G$}{\If{$goal_i\notin D_G$}{
$goals_2\leftarrow goals_2 goal_i$\;
 $D_G \leftarrow D_G \cup \{goal_i\}$
}}
}

\KwRet $vGOALSpec_2$
\end{algorithm}

Algorithm \ref{Alg:3} simplifies a \emph{vGOAL} specification for a one-agent system. The given input is a \emph{vGOAL} specification for a one-agent system. The agent has many goals. The output is also a \emph{vGOAL} specification. The difference between the input and output is only the goal specification.  $goals_1$ is the goal specification of the input, and $goals_2$ is the goal specification of the output. $Goals$ stores different goals. Each $goal_i$ is an element of $Goals$.

\begin{algorithm}
  \caption{Detect errors of a \emph{vGOAL} specification for a $n$-agent system with multiple goals}
  \label{Alg:4}
\LinesNumbered
\SetKwInput{KwInput}{Input}                
\SetKwInput{KwOutput}{Output}              
\DontPrintSemicolon
\KwInput{$vGOALSpec_1::= (MAS_1,K, C, A,S, P,E, D)$, $MAS::=(i, B, goals_i, M_S, M_R)^n$, $goals_{i}::= goal_{i1} goal_{i2}...goal_{in_i}$, $1\le i\le n_i$,
$Goals::=\{G_1,...,G_m\}$ $\forall 1\le j\le n_i. goal_{ij}\in Goals$, $1\le n$, $1\le m$, $P_{safe}$}
  \KwOutput{$Errors$}
  $D\leftarrow \bigcup_{i=1}^n\{i:\emptyset\}$, $Errors\leftarrow \emptyset$\;
\For{each $MAS_i\in MAS$ }{
$Spec\leftarrow (MAS_i, K, C, A,S, P,E, D)$\;
$Spec_1\leftarrow Algorithm_3(Spec)$
\;
$TS_1\leftarrow Algorithm_1(vGOALSpec)$\;
\If{$TS_1\not\models P_{safe}$}{
$Errors\leftarrow Errors\cup \{(i,TS_1\not\models P_{safe})\}$
}
\For{each $goal\in goals$}{
\If{$TS_1\not\models \lozenge goal$}{
    $Errors\leftarrow Errors\cup \{(i,TS_1\not\models \lozenge goal)\}$
} 
}

}
$Goals\leftarrow \Pi_{i=1}^n|D_G[i]|$\;
\For{each $(g_1,...,g_n)\in Goals$}{
$vGOALSpec_2\leftarrow vGOALSpec_1$\;
\For{each $agent_i\in MAS$}{
$goals_i\leftarrow g_i$
}
$TS_n\leftarrow Algorithm_1(vGOALSpec_2)$\;
\If{$TS_n\not\models P_{safe}$}{
$Errors\leftarrow Errors\cup \{((g_1,...,g_n), TS_n\not\models P_{safe})\}$
}
\If{$TS_n\not\models \lozenge \bigwedge_{i=1}^n goal_i$
}{
$Errors\leftarrow Errors\cup \{TS_n\not\models \lozenge \bigwedge_{i=1}^n g_i\}$
}

}
\KwRet $Errors$
\end{algorithm}

Algorithm \ref{Alg:4} detects errors of a \emph{vGOAL} specification for a $n$-agent system that each agent has multiple goals by model checking safety properties and liveness properties. First, it employs Algorithm \ref{Alg:3} to generate a set of \emph{vGOAL} specifications for a one-agent system. Algorithm \ref{Alg:1} transforms each generated specification to a transition system. 
Second, it generates a set of \emph{vGOAL} specifications for a $n$-agent system, in which each agent has one goal. The generated \emph{vGOAL} specifications for a $n$-agent system cover all possibilities of the goal combination of all agents within the autonomous systems. Safety checking is applied to each generated transition system. Liveness checking is used to check if all agent goals will be finally achieved. $Errors$ records all property violations during all model-checking processes.

\subsection{Theoretical Proofs}
\begin{theorem}\label{Thm:one-agent}Given a \emph{vGOAL} specification for a one-agent system, $vGOALSpec_{1}$, where $goals_1::=goal_1...goal_n$,  Algorithm \ref{Alg:3} reduces $vGOALSpec_{1}$ to $vGOALSpec_{2}$, where $goals_1::=goal_1...goal_k$.
Algorithm \ref{Alg:1} generates a transition systems $TS_1$ with the state space $S_1$ for $vGOALSpec_{1}$, and $TS_2$ with the state space $S_2$ for $vGOALSpec_{2}$. $S_1=S_2$. 
\end{theorem}
\begin{proof} We prove the theorem by induction on $n$. \\
\textbf{Base case (n=1)} $S_{10}::=(I(B),I(goals_1))$, $S_{20}::=(I(B),$ $I(goals_2))$, $I(goals_1)=I(goals_2)$. Algorithm \ref{Alg:1} will go through the same transition system generation process for the same initial state until the current focused goal is shifted from the initial state. Therefore,  $S_1= S_2$. \\
\textbf{Induction Step} Assume $n\ge 0$, $S_1= S_2$. For the $n+1$-th goal of $goals_1$, if $\forall 2\le i\le n. goal_i\ne goal_{n+1}$, $goals_2$ is also extended by $goal_{n+1}$. Algorithm \ref{Alg:1} will go through the same transition system generation process for the same state until the current focused goal is shifted. Therefore,  $S_1= S_2$. If $\exists 2\le i\le n. goal_i= goal_{n+1}$, $goals_2$ will not be extended by $goal_{n+1}$. Algorithm \ref{Alg:1} will extend $TS_1$ to achieve $goal_{n+1}$, but the generation process for $goal_{n+1}$ is the same as the generation process for $goal_{i}$. No new states will be generated in this case. $S_1$ will not be modified by the generation process. Therefore, $S_1=S_2$.
\end{proof}

\begin{theorem} \label{thm:Error}
Given a \emph{vGOAL} specification for a $n$-agent system with multiple goals, $vGOALSpec$, Algorithm \ref{Alg:4} processes it to generate $Errors$ as the output. $Errors\ne \emptyset$ implies $vGOALSpec$ is not a sound specification that can successively achieve all goals of each agent without violating all safety properties. 
\end{theorem}
\begin{proof} $Errors\ne \emptyset$ implies at least one error is detected during the model-checking process. If $(i,TS_1\not\models P_{safe})\in Errors$,
$vGOALSpec$ will generate unsafe decisions when successively achieving all goals of $agent_i$. If $(i,TS_1\not\models \lozenge goal)\in Erros$, $vGOALSpec$ cannot generate decisions to achieve $goal$ for $agent_i$. If $\{((g_1,...,g_n), TS_n$ $\not\models P_{safe})\}\in Errors$, $vGOALSpec$ will generate unsafe decisions when generating decisions to simultaneously to achieve $g_i$ for $agent_i$. If $\{TS_n$ $\not\models \lozenge \bigwedge_{i=1}^n g_i\}\in Errors$, $vGOALSpec$ cannot generate decisions for an $n$-agent system to achieve each goal of each agent if the goals combination is $(g_1,...,g_n)$. Any error of these four types makes the \emph{vGOAL specification} not sound.
\end{proof}

\section{Empirical Analysis}
This section illustrates the efficiency and the usage of the automated CTL and PCTL model-checking process of \emph{vGOAL}. First, we briefly introduce the case study and its scenarios for experiments. Second, we present how to use the CTL and PCTL model-checking process to verify a \emph{vGOAL} specification. Finally, we show how to quickly detect errors in a complicated \emph{vGOAL} specification. 
We use a real-world autonomous logistic system, involving three autonomous mobile robots, to conduct an empirical analysis. All experiments are conducted with a MacBook Air 2020 with an Apple M1 and 16GB of RAM. All \emph{vGOAL} specifications and experiment results described in this section are available at \cite{vGOALsource}.

\subsection{Case Study}
The autonomous logistic system is expected to work in a warehouse environment, shown in Figure \ref{fig:layout}. It aims to autonomously and continuously transport workpieces from two pick-up stations ($P_3$ and $P_4$) to one drop-off station ($P_2$) while adhering to all safe requirements. Each agent can have many transportation goals, but it only can have two different transportation goals, namely transporting a workpiece from $P_3$ to $P_2$ or from $P_4$ to $P_2$. Three agents collaboratively deliver workpieces, but they also compete for permission for locations to avoid collision among agents. 

\begin{figure}[h]
    \centering
    \includegraphics[scale=0.5]{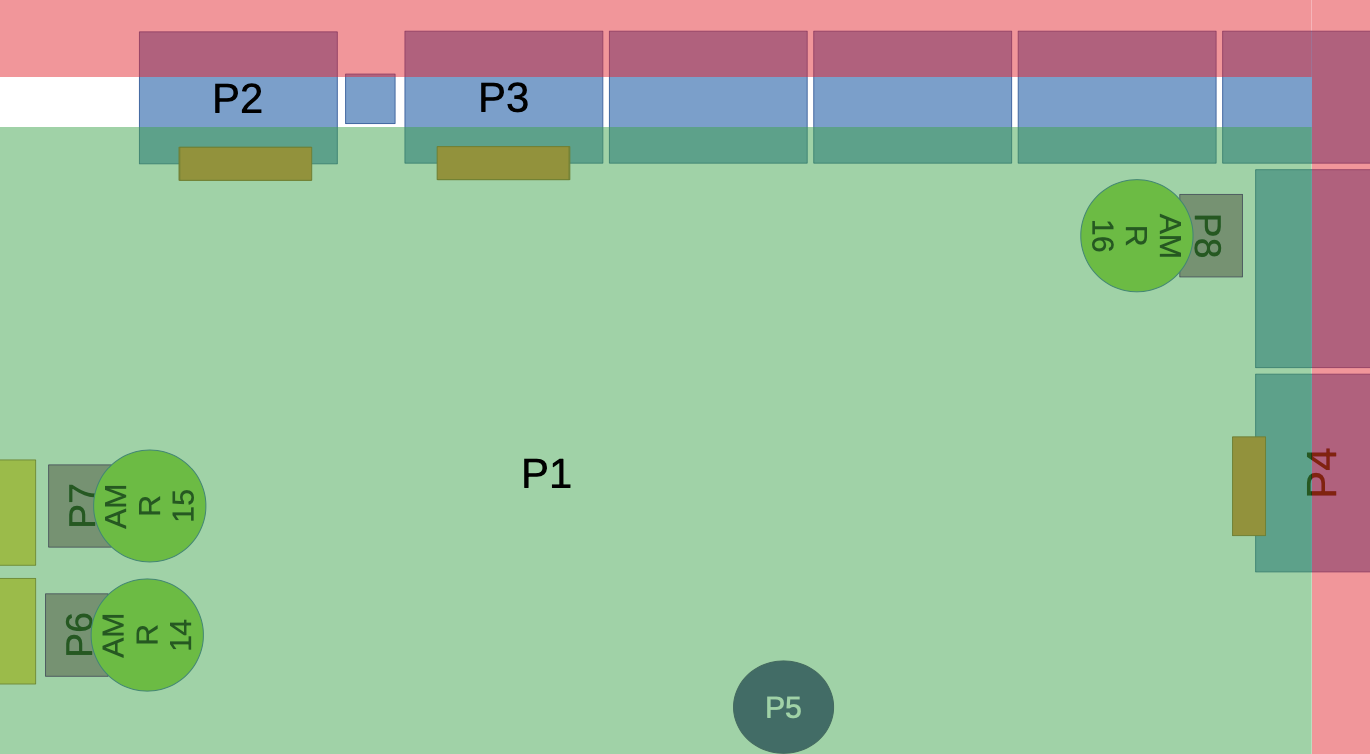}
    \caption{Layout of the Warehouse Environment}
    \label{fig:layout}
\end{figure}

In this case study, each agent can perform four actions: $pick$, $drop$, $charge$, and $move$. Each action can succeed or fail. For the two actions: $pick$ and $drop$, and $charge$, their failure is not fatal for the agent. For the $charge$ action, its failure will lead to the failure of the corresponding agent. For the $move$ action, its failure can have possible consequences. The first possible consequence is a docking error, which is not fatal to the agent. The second possible consequence is a move-base error, which leads to a system crash. 

We use \emph{vGOAL} to specify the decision-making mechanism of the autonomous logistic systems. The agent specifications are listed as follows: 
\begin{lstlisting}[xleftmargin=2.0ex]
beliefs1 = ["at(6)", "battery(1)", "docked(6)", "assigned(6)"]
beliefs2 = ["at(7)", "battery(1)","docked(7)","assigned(7)"]
beliefs3 = ["at(8)", "battery(1)","docked(8)","assigned(8)"]
beliefs4 = ["idle(2)", "idle(3)", "idle(4)","idle(5)", 
"reserved(A1,6)", "reserved(A2,7)", "reserved(A3,8)"]
g1 = ['transport(3,2)']
g2 = ['transport(4,2)']
goals4 = []
Agent1 = Agent("A1", beliefs1, goals1)
Agent2 = Agent("A2", beliefs2, goals2)
Agent3 = Agent("A3", beliefs3, goals3)
Resource = Agent("R", beliefs4, goals4)
System=[A1, A2, A3, R]
\end{lstlisting}

$System$ specifies all agents involved in the specified autonomous system. $Agent1$ specifies the autonomous mobile robot locating at $P_6$, $Agent2$ specifies the autonomous mobile robot locating at $P_7$, and $Agent3$ specifies the autonomous mobile robot locating at $P_8$. $Resource$ specifies a resource agent that manages the critical resources.  $g_1$ and $g_2$ are two delivery goals. $goals_1$, $goals_2$, and $goals_3$ can contain many goals, but each goal can be either $g_1$ or $g_2$, e.g., $goals1=[g_1, g_1, g_2, g_2, g_1]$. 

The autonomous logistic system is expected to always satisfy two safety requirements: safe battery levels and safe locations. The related safety specifications are specified in \emph{vGOAL} as follows:
\begin{lstlisting}[xleftmargin=2.0ex]
"battery(1) implies safe1",
"battery(2) implies safe1",
"exists p. at(p) and not at(9) implies safe2",
safety_properties = {"A1": ["safe1", "safe2"], "A2": ["safe1","safe2"], "A3": ["safe1","safe2"]}
\end{lstlisting}

\subsection{Model-Checking Analyses}
Given a \emph{vGOAL} specification, the automated CTL and PCTL model-checking processes will automatically generate the CTL and PCTL analyses. Obviously, CTL and PCTL are not semantically equivalent. Hence, they can express different properties, and we can use CTL and PCTL model-checking processes to verify different properties. However, CTL and PCTL can also equivalently express some qualitative properties, and we can use them to compare the efficiency of CTL and PCTL model-checking processes for \emph{vGOAL}. 

To allow the efficiency comparison of these two model-checking analyses, we define two qualitative properties that can be equivalently expressed in both CTL and PCTL. Consequently, these properties have the same satisfaction results for a given \emph{vGOAL} specification.

First, we verify there exists a path where no errors happen and it finally reaches a final state, which can be equally formalized in CTL and PCTL as follows:
\begin{itemize}
    \item CTLSPEC EG (\emph{non-errors}$\rightarrow$ EF (\emph{liveness})),
    \item $P>0$ [\emph{non-errors} U \emph{non-errors} $\land$ \emph{liveness}].
\end{itemize}
Both properties have the same satisfaction results. 

Second, safety properties can be equally formalized in CTL and PCTL as follows:
\begin{itemize}
    \item CTL in NuSMV: CTLSPEC AG \emph{safety}.
    \item PCTL in Storm: $P>=1$ [G \emph{safety}].
\end{itemize}
\label{equivprop}

\begin{table}[h]
 \caption{Efficiency of Model Generation and Model Checking}
    \label{tab:ETS}
    \centering
  \begin{tabular}{?c|c|c|c|c?c|c?c|c?} 
 \hline
 N &Agent 1 &Agent 2& Agent 3& States & \makecell{NuSMV \\Input (s)} &\makecell{Storm \\Input(s)}&\makecell{NuSMV \\Check(s)}&\makecell{Storm \\Check(s)}\\ [0.6ex] 
 \hline
   \hline
3& $g_1$  &$g_1$&$g_1$ &1799&90.19& 89.37&1.70& 6.96\\
\hline
3& $g_1$  &$g_2$&$g_1$ &3338&213.38 &213.31 &4.75& 13.44\\
\hline
3& $g_1,g_1$  &$g_2$&$g_1$ &5923&559.16 & 558.18& 13.28&29.76
 \\
\hline
3& $g_1,g_1$  &$g_1,g_1$&$g_2$ & 9975 & 1398.40&1400.85& 38.42&80.02\\

 \hline
3& $g_1, g_1$  &$g_2, g_2$&$g_1$ &10598 &1581.14&1583.39&59.13&87.10\\
\hline
3& $g_1, g_1$  &$g_1, g_1$&$g_1,g_1$ & 8484&1067.45 &  1069.13& 27.94&61.68\\
 \hline
 3& $g_1, g_1$ &$g_2, g_2$ &$g_1, g_1$ &16139&3321.36 &3318.55&95.27&181.55\\
 \hline
3&  $g_1, g_2$  &$g_2, g_1 $&$g_1, g_2$&22572&6320.08&6362.64 &224.78 &348.48\\
 \hline
3& $g_1, g_1, g_1$  &$g_2, g_2, g_2$&$g_1, g_1, g_1$ &22652&6408.20&6687.92 &220.19 &322.29\\
\hline
3& $g_1, g_2, g_1$  &$g_2, g_1, g_1$&$g_1, g_2, g_1$ &45743&25022.82&25008.40&1461.47&2153.32\\
 \hline
\end{tabular}
\end{table}

Table \ref{tab:ETS} presents the time required for generating the input models and performing model checking for 27 representative scenarios of the case study. $N$ represents the number of agents included in the system. $Agent 1$, $Agent 2$, and $Agent 3$ represent three AMRs that can be part of the system, whose value represents its goals. $States$ shows the number of states of the generated model for the given \emph{vGOAL} specification.

$NuSMV Input$ and $Storm Input$ represent the time taken from the given \emph{vGOAL} specification to the inputs of NuSMV and PRISM, respectively. 
The time for generating the inputs for both NuSMV and PRISM is almost the same.

$NuSMV Check$ and $Storm Check$ represent the time taken to verify the two qualitative properties of NuSMV and PRISM, respectively. NuSMV generally shows slightly faster to verify these properties compared with PRISM across all scenarios. However, the magnitudes of their efficiency are still the same. We notice that the relative difference in the model checking between the NuSMV and Storm becomes smaller with the increase in the complexity of the system. 

The efficiency of the CTL and PCTL model-checking for \emph{vGOAL} depends on the efficiency of the model generation and the model checking. Based on the results shown in Table \ref{tab:ETS}, the efficiency of the model generation is almost the same for both model-checking analyses. Moreover, NuSMV is slightly faster than Storm but within the same magnitude. 

Results in Table \ref{tab:ETS} reflects the bottleneck of model checking: state-space explosion. For example, the total state space increases dramatically from 1,799 to 45,743 states as the goals for each agent expand from one to three. Consequently, the total time for the model-checking processes increases from less than 2 minutes to more than 7 hours. 

In autonomous systems, each agent typically has multiple goals. When dealing with complex systems consisting of numerous agents, each with many goals, the time required for model checking can become substantial. Thus, it is crucial and practical to quickly detect errors for a complicated \emph{vGOAL} specification. Instead of directly verifying the whole state space for such a complicated system, we first choose to efficiently verify the most properties such as safety and liveness of the complicated system using Algorithm \ref{Alg:4}. 

\begin{table}[h]
 \caption{Efficiency of CTL Model Checking Using Quick Error Detection Algorithm}
    \label{tab:quick}
    \centering
  \begin{tabular}{|c|c|c|c|c|c|c| } 
 \hline
 N &Agent 1 &Agent 2& Agent 3& States &  \makecell{NuSMV \\Input (s)} &\makecell{NuSMV \\Check(s)}\\ [0.6ex] 
 \hline
   \hline
 1  & $g_1,g_2,g_1$ & \diagbox{}&\diagbox{} &142&4.26 & 0.09\\
  \hline
 1  &\diagbox{}& $g_2,g_1, g_1$  & \diagbox{}&125& 3.68&0.08 \\
 \hline
 1  &\diagbox{}&\diagbox{}& $g_1, g_2, g_1$  &142&4.24 &0.08\\
 \hline
3 & $g_1$  &$g_1$&$g_1$ &1799 &90.19&1.70 \\
\hline
3  & $g_1$  &$g_2$&$g_1$ &3338 &213.38&4.75 \\
\hline
3  & $g_1$  &$g_1$&$g_2$ & 3660&240.56&5.694 \\
 \hline
3  & $g_1$  &$g_2$&$g_2$ &2976 &179.33&4.22 \\
\hline
3  & $g_2$  &$g_1$&$g_1$ & 2976&181.05&4.26 \\
\hline
3  & $g_2$  &$g_2$&$g_1$ & 3660&242.06& 5.77\\
\hline
3  & $g_2$  &$g_1$&$g_2$ & 3338&211.61&4.84 \\
\hline
3  & $g_2$  &$g_2$&$g_2$ & 1799&90.68&2.07 \\
 \hline
\hline
Total&\diagbox{}&\diagbox{}&\diagbox{}&24355&1461.62&39.51\\
\hline
\end{tabular}
\end{table}

We use the last scenario presented in Table \ref{tab:ETS} as the example to illustrate the usage and efficiency of Algorithm \ref{Alg:4} for quick error detection. As shown in the last row of Table \ref{tab:ETS}, the total states is 45743, the time for generating NuSMV inputs and NuSMV checking is 25022.82 seconds and 1461.47 seconds, respectively.  

For the last scenario, Algorithm \ref{Alg:4} generates three one-agent autonomous systems and eight three-agent autonomous systems. We use the automated CTL model-checking processes to obtain 11 CTL model-checking analyses for safety and liveness properties. The total verified states of the 11 model-checking analyses is 24355. The total time for generating NuSMV inputs and NuSMV checking of the 11 model-checking analyses is 1461.62 seconds and  39.51 seconds, respectively. 

Notably, the total verified states across these 11 analyses are reduced from 45743 to 24355. However,  the time required for generating NuSMV inputs decreases significantly from to 25022.82 to 1461.62 seconds. Moreover, the NuSMV checking process is further drastically reduced from 1461.47 to 39.51 seconds. These results demonstrate that we can use Algorithm \ref{Alg:4} to efficiently detect errors of \emph{vGOAL} specifications, particularly in complex scenarios.

\section{Conclusion}
This paper presents CTL and PCTL automated model-checking processes for \emph{vGOAL}, enabling obtaining sound \emph{vGOAL} specifications for complex and real-world autonomous decision-making. Compared with other model-checking approaches of APLs, the presented model-checking process of \emph{vGOAL} has three main advantages. First, it has an efficient model-checking process. The experiment results demonstrate that the CTL and PCTL model-checking processes for \emph{vGOAL} can efficiently verify the \emph{vGOAL} specifications in real-world scenarios. Second, the semantic equivalence is established for both nondeterministic models and probabilistic models of \emph{vGOAL}: from \emph{vGOAL} to transition systems or DTMCs. 
Third, Algorithm \ref{Alg:4} is proposed to efficiently verify whether a \emph{vGOAL} specification violates safety or liveness properties. This algorithm is especially useful for complex scenarios. We plan to conduct an empirical analysis for similar approaches, like model checking Gwendolen using MCAPL.

\end{document}